\documentclass[11pt]{article}
\usepackage[titletoc,title]{appendix}
\usepackage{amsmath,amsfonts,amssymb,bm,amsthm}
\usepackage{etoolbox}
\makeatletter
\patchcmd{\maketitle}
  {\renewcommand\thefootnote{\@fnsymbol\c@footnote}}
  {\renewcommand\thefootnote{\arabic{footnote}}}
  {}{\typeout{Failed to patch maketitle footnote numbering}}
\makeatother
\newcounter{affPavia}
\newcounter{affINFN}
\newcounter{affINDAM}
\usepackage[colorlinks,citecolor=blue,linkcolor=blue,pagebackref]{hyperref}
\usepackage{tikz}
\usetikzlibrary{shapes,arrows}
\usetikzlibrary{intersections,shapes.arrows}
\usetikzlibrary{decorations.pathreplacing}
\usetikzlibrary{arrows.meta}
\numberwithin{equation}{section}
\newtheorem{theorem}{Theorem}[section]

\newtheorem{proposition}[theorem]{Proposition}
\newtheorem{corollary}[theorem]{Corollary}

\theoremstyle{definition}

\newtheorem{definition}[theorem]{Definition}

\theoremstyle{remark}
\newtheorem{remark}[theorem]{Remark}

\usepackage{multicol}
\usepackage{longtable}
\usepackage{slashed}
\usepackage{accents}
\usepackage[T1]{fontenc}
\usepackage[top=1in, bottom=1in, left=1in, right=1in]{geometry}
\usepackage{fancyhdr}
\newcommand{\dr}{\mathrm{d}}
\newcommand{\dist}{\operatorname{dist}}
\newcommand{\Ric}{\operatorname{Ric}}
\newcommand{\Riem}{\operatorname{Riem}}
\newcommand{\Sc}{\operatorname{Sc}}

\newcommand{\prin}{\mathrm{prin}}
\newcommand{\sub}{\mathrm{sub}}
\newcommand{\loc}{\mathrm{loc}}

\definecolor{amber}{rgb}{1.0, 0.49, 0.0}

\usepackage{enumerate}
\usepackage{xcolor,cancel}
\usepackage{relsize}

\renewcommand{\tilde}{\widetilde}

\usepackage{orcidlink}

\allowdisplaybreaks

\usepackage{pifont}
\renewcommand*{\backrefalt}[4]{%
\ifcase #1 %
No citations%
\or
\ding{43}~p.~#2%
\else
\ding{43}~pp.~#2%
\fi}
\usepackage{soul}

\begin{document}

\title{Spectral asymmetry via pseudodifferential projections:\\ the massless Dirac operator}
\author{Matteo Capoferri\,\orcidlink{0000-0001-6226-1407}%
\thanks{Dipartimento di Matematica ``Federigo Enriques'',
Universit\`a degli Studi di Milano,
Via C.~Saldini 50,
20133 Milano,
Italy \emph{and}
Maxwell Institute for Mathematical Sciences, Edinburgh \&
Department of Mathematics,
Heriot-Watt University,
Edinburgh EH14 4AS,
UK;
\texttt{matteo.capoferri@unimi.it},
\url{https://mcapoferri.com}.}
\and
Beatrice Costeri%
\thanks{Dipartimento di Fisica ``Alessandro Volta'',
Universit\`a degli Studi di Pavia,
Via Bassi 6,
I-27100 Pavia,
Italia;
\texttt{beatrice.costeri01@universitadipavia.it},
\texttt{claudio.dappiaggi@unipv.it}.}\,
\setcounter{affPavia}{\value{footnote}}%
\thanks{INFN, Sezione di Pavia,
Via Bassi 6,
I-27100 Pavia,
Italia.}\,
\setcounter{affINFN}{\value{footnote}}%
\thanks{INdAM, Sezione di Pavia,
Via Ferrata 5,
I-27100 Pavia,
Italia.}%
\setcounter{affINDAM}{\value{footnote}}%
\and
Claudio Dappiaggi\,\orcidlink{0000-0002-3315-1273}%
\footnotemark[\value{affPavia}]\,
\footnotemark[\value{affINFN}]\,
\footnotemark[\value{affINDAM}]%
}


\date{}

\maketitle

\vspace{-.6cm}

\begin{abstract}
A new approach to the study of spectral asymmetry for systems of partial differential equations (PDEs) on closed manifolds was proposed in a recent series of papers by the first author and collaborator. They showed that information on spectral asymmetry can be encoded within and recovered from  a negative order pseudodifferential operator --- the asymmetry operator --- constructed from appropriately defined pseudodifferential (spectral) projections. 
In this manuscript we apply these techniques to the study of the massless Dirac operator; in particular, we compute the principal symbol of the asymmetry operator, accounting for the underlying gauge invariance.


\

{\bf Keywords:} Dirac operator, spectral asymmetry, eta invariant, pseudodifferential projections.

\

{\bf 2020 MSC classes: }
primary
58J50; 
secondary
35P20, 
35Q41, 
47B93, 
58J28, 
58J40. 

\end{abstract}

\tableofcontents

\allowdisplaybreaks

\section{Statement of the problem and main results}
\label{Statement of the problem and main results}

Let $(M,g)$ be a closed, connected and oriented Riemannian 3-manifold. 
To fix our convention, let us specify that the Riemann curvature tensor $\Riem$ locally has components ${\Riem^\kappa}_{\lambda\mu\nu}$ defined in accordance with 
\[
{\Riem^\kappa}_{\lambda\mu\nu}:=
\dr x^\kappa(\Riem(\partial_\mu\,,\partial_\nu)\,\partial_\lambda)
=
\partial_\mu{\Gamma^\kappa}_{\nu\lambda}
-\partial_\nu{\Gamma^\kappa}_{\mu\lambda}
+{\Gamma^\kappa}_{\mu\eta}{\Gamma^\eta}_{\nu\lambda}
-{\Gamma^\kappa}_{\nu\eta}{\Gamma^\eta}_{\mu\lambda}\,,
\]
the $\Gamma$'s being Christoffel symbols. The Ricci tensor is defined as $\Ric_{\mu\nu}:=\Riem^\alpha{}_{\mu\alpha\nu}$ and $\Sc:=g^{\mu\nu}\Ric_{\mu\nu}$ is scalar curvature.

We denote by $\text{L}^2(M)$ the linear space of 2-columns of half-densities over $M$ --- that is, sections of the trivial $\mathbb{C}^2$-bundle over $M$ valued in half-densities --- and by $\text{H}^1(M)$ the Sobolev space consisting of 2-columns of half-densities $f \in \text{L}^2(M)$ with first order weak derivatives also in $\text{L}^2(M)$. 

For $m \in \mathbb{R}$, we define $\Psi^m(M)$ to be the \emph{space of classical pseudodifferential operators of order $m$} acting on two-columns of half-densities, and $S^m(M)$ to be the space of polyhomogeneous symbols of order $m$. 
In addition, following standard conventions, we denote by
\begin{equation}
	\label{smoothop}
	\Psi^{-\infty}(M) := \bigcap_{m \in \mathbb{R}} \Psi^m(M)
\end{equation}
the space of integral operators with infinitely smooth integral kernel.

Given a pseudodifferential operator $P \in \Psi^m(M)$, we shall denote by $P_{\text{prin}}$ its principal symbol and by $P_{\text{sub}}$ its subprincipal symbol \cite[Eqn. (5.2.8)]{DuHo}. When needed to avoid confusion, we denote by $(\,\cdot\,)_{\text{prin}, m}$ the principal symbol of the quantity within brackets, regarded as a pseudodifferential operator in $\Psi^{-m} (M)$. Furthermore, we shall denote by $p(x,\xi)$ (lowercase, roman) the full symbol of $P$, and by $\mathfrak{p}(x,y)$ (lowercase, fraktur) its Schwartz (integral) kernel. For further details and background materials on the spaces and properties of pseudodifferential operators we refer the reader to~\cite{shubin}.

\

As a consequence of our assumptions, $M$ is \emph{parallelisable} \cite{Stiefel}, namely there exist globally defined smooth vector fields $\{e_j\}_{j =1}^3 \in \Gamma(T M)$ such that, at each point $x \in M$, $\{e_j(x)\}_{j =1}^3$ is a basis of $T_x M$.
Working in a local coordinate patch, we denote by $e_j{}^{\alpha}$, $\alpha = 1,2,3$, the $\alpha$-th component of the $j$-th vector field, and by 
\begin{equation}
	\label{pauliproj}
	\sigma^{\alpha} (x) := s^j e_j{}^{\alpha}(x)
\end{equation}
the projections of the standard Pauli matrices $\{s^j\}_{j=1,2,3}$ along the chosen global framing. Here and further on we adopt Einstein's summation convention over repeated indices; we shall use Greek letters for holonomic (tensor) indices and Latin for anholonomic (frame) indices.

Our manifold $M$ admits a spin structure, see~\cite[Ex. 2.3]{lawson}, which in turn entails the existence of a  \emph{(massless) Dirac operator}.

\begin{definition}
\label{masslessdirop}
We call \emph{massless Dirac operator} the operator acting on $2$-columns of half-densities defined as
\begin{equation}
\label{dirop}
    W := -i \sigma^{\alpha} \left( \frac{\partial}{\partial x^{\alpha}} + \frac{1}{4} \sigma_{\beta} \left( \frac{\partial \sigma^{\beta}}{\partial x^{\alpha}} + \Gamma^{\beta}{}_{\alpha\gamma} \sigma^{\gamma} \right) - \frac{1}{2} \Gamma^{\beta}{}_{\alpha\beta}\right): H^1(M) \rightarrow \text{L}^2(M)\,.
\end{equation}
\end{definition}

\begin{remark}
The reader may be more accustomed with the ``classical'' massless Dirac operator $W_{\mathrm{scal}}$ acting on $2$-columns of scalar fields, \textit{i.e.}, \emph{Weyl spinors}. The latter can be obtained from $W$ by conjugation with an appropriate power of the Riemannian density $\rho(x) := \sqrt{\det g_{\mu\nu}(x)}$: 
\begin{equation}
    W_{\text{scal}}:= \rho^{\frac{1}{2}} W \rho^{-\frac{1}{2}}. 
\end{equation}
Let us also mention that the massless Dirac operator is sometimes referred to as \emph{Weyl operator} in the literature --- this justifies the notation $W$ in Definition~\ref{masslessdirop}. 
\end{remark}

\

The massless Dirac operator is a self-adjoint first order pseudodifferential operator whose principal symbol, a smooth matrix-function on $T^*M\setminus\{0\}$ positively homogeneous in momentum $\xi$ of degree $1$, reads
\begin{equation}
    \label{Wprin}
    W_{\prin}(x, \xi) = \sigma^{\alpha}(x) \,\xi_{\alpha}.
\end{equation}
The principal symbol $W_{\prin}$ has \textit{simple} eigenvalues
\begin{equation}
	h^{(\pm)} (x, \xi) = \pm\|\xi\| =: \pm h(x, \xi)\,,
\end{equation}
where $\|\xi\|^2 := g^{\mu\nu}(x) \,\xi_{\mu} \xi_{\nu}$ is the Riemannian norm squared of $\xi$.
Hence, the massless Dirac operator $W$ is elliptic, since the determinant of its principal symbol
\begin{equation}
	\det W_{\prin}(x, \xi) = - \|\xi\|^2
\end{equation}
is nowhere vanishing on the punctured cotangent bundle $T^*M\setminus \{0\}$.
This implies that the spectrum of $W$ is discrete and the operator is unbounded both from above and below, with eigenvalues accumulating to $+\infty$ and $-\infty$.

The remaining part of the full symbol of $W$, positively homogeneous in momentum $\xi$ of degree zero, reads
\begin{equation*}
    W_0 := -i \sigma^{\alpha} \left( \frac{1}{4} \sigma_{\beta} \left( \frac{\partial \sigma^{\beta}}{\partial x^{\alpha}} + \Gamma^{\beta}{}_{\alpha\gamma} \sigma^{\gamma} \right) - \frac{1}{2} \Gamma^{\beta}{}_{\alpha\beta}\right).
\end{equation*}

Since we are interested in the study of the spectral problem for the massless Dirac operator, inspired by \cite{part1} we introduce the following quantities.

\begin{definition}\label{Def: Spectral Projectors}
	Denoting by $\lambda_k$ the non vanishing eigenvalues of the massless Dirac operator $W$ (with account of their multiplicities) and by $u_k$ the corresponding eigenspinors, we define
	\begin{flalign}
		\label{P_+}
		P_+ &:= \sum_{\lambda_k > 0} \langle u_k, \cdot \rangle u_k, \\
        \label{P_-}
		P_- &:= \sum_{\lambda_k 
			< 0} \langle u_k, \cdot \rangle u_k,
	\end{flalign}
    to be the positive and negative spectral projections,
	where $\langle v , w \rangle:=\int_M \overline{v^T(x)}\,w(x) \,\mathrm{d}x$ denotes the natural inner product in $\text{L}^2(M)$, while the superscript $T$ stands for the transpose. Furthermore, we denote by
	\begin{equation}
		\label{P_0}
		P_0 := \sum_{j=1}^{N} \langle v_j, \cdot \rangle \, v_j,
	\end{equation}
the projection onto the kernel of $W$, where $N := \dim \ker W < +\infty$ is the multiplicity of the eigenvalue zero which is finite, because $W$ is elliptic.
\end{definition}

Observe that, for notational convenience, we distinguish the eigenspinors associated with the eigenvalue zero denoting them by $v_j$.
Furthermore $v_j \in C^{\infty}(M)$ for all $j\in\{1,\dots,N\}$, due to the ellipticity of $W$, which entails that $P_0$ is a smoothing operator. 

The operators $P_\pm$ and $P_0$ are pseudodifferential operators of order zero \cite{part1}, satisfying, for $\aleph,\beth\in\{+,-, 0\}$, 
\begin{flalign}
    P_{\aleph}^2 &= P_\aleph \\
    P_\aleph^* &= P_\aleph, \\
    P_\aleph P_\beth &= 0  \quad \text{for } \aleph \ne \beth \\
    P_++P_-  &= \text{Id} \mod \Psi^{-\infty}(M)\,,
\end{flalign}
where $\text{Id} \in \Psi^0(M)$ denotes the identity operator. 

\vskip .3cm

If $Q \in \Psi^m(M)$ is self-adjoint and $m < -3$ (recall that the dimension of our manifold is $3$), then $Q$ is of trace class and, in addition, $Q$ is an integral operator with continuous integral kernel \cite[Sec. 12.1]{shubin}. In other words, we can represent the action of $Q$ as 
\begin{equation}
\label{pstrace}
    Q: f_j(x) \mapsto \int_{M} \mathfrak{q}_j{}^{k}(x,y) f_k(y) \, \mathrm{d}y, \qquad j = 1, 2, 3,
\end{equation}
where $\mathrm{d}y =  \mathrm{d}y^1  \mathrm{d}y^2  \mathrm{d}y^3$. 

\

Our construction will rely on the notion of \emph{matrix trace} of a matrix pseudodifferential operator --- see also~\cite[Section~4]{curl}, \cite[Section~3]{conjectures}.

\begin{definition}\label{Def: Trace of a Pseudo}
	Given $Q\in\Psi^m(M)$, we call the \emph{matrix trace of the pseudodifferential operator $Q$} the scalar pseudodifferential operator
	\begin{equation}
		\label{matrixtr}
		\mathfrak{tr} \, Q : f(x) \mapsto \int_{M}  (\text{tr} \, \mathfrak{q})(x,y) f(y) \, \mathrm{d}y,
	\end{equation}
	where
	\begin{equation}
		\label{trp}
		(\text{tr} \, \mathfrak{q})(x, y) := \mathfrak{q}_j{}^{j}(x,y)\,.
	\end{equation}
\end{definition}

\begin{remark}
    Note that the operator defined by equation \eqref{matrixtr} is well-defined for \emph{any} $m \in \mathbb{R}$, not necessarily satisfying the condition $m < -3$. Indeed, when $m \ge -3$, the integral kernel $\mathfrak{q}_j{}^k(x, y)$ in equation~\eqref{pstrace} should be interpreted at a distributional level, \textit{i.e.}, as a Schwartz kernel. However, the operation of taking its matrix trace is still well-defined. 
\end{remark}

\noindent The matrix trace satisfies the following properties \cite[Section~4]{curl}:
\begin{itemize}
    \item[(i)] if $Q\in\Psi^m(M)$, then $\mathfrak{tr} \, Q$ is a scalar pseudodifferential operator of order $m$;
    \item[(ii)] 
       $ (\mathfrak{tr} \, Q)^* = \mathfrak{tr} (Q^*)$,
    where $Q^*$ the formal adjoint of $Q$ with respect to inner product in $L^2(M)$;
    \item[(iii)] 
       $ (\mathfrak{tr} \, Q)_{\text{prin}} = \text{tr} (Q_{\text{prin}})$;
    \item[(iv)] $(\mathfrak{tr} \, Q)_{\text{sub}} = \text{tr} (Q_{\text{sub}})$;
    \item[(v)] if $Q$ is of trace class, then
    \begin{equation}
\label{optrace}
    \text{Tr} \, Q = \text{Tr} (\mathfrak{tr} \,Q)\,,
\end{equation}
where $\operatorname{Tr}Q$ denotes the operator trace of $Q$.
\end{itemize}

The approach developed in \cite{curl,conjectures} is rooted in the the attempt to define spectral asymmetry ``directly'', as the trace of the difference of positive and negative spectral projections. The issue one encounters is that, clearly, $P_+-P_-$ is not of trace class; indeed, $+1$ and $-1$ are points of essential spectrum --- eigenvalues of infinite multiplicity. As an intermediate step, what one can do is to consider the matrix trace of $P_+-P_-$.

\begin{definition}
\label{asyop}
    We call \emph{asymmetry operator} the pseudodifferential operator 
    \begin{equation}
    \label{definition asymmetry operator equation}
        A := \mathfrak{tr} (P_+ - P_-)\,.
    \end{equation}
\end{definition}

\emph{Prima facie}, the operator $A$ is of order zero. It was shown in~\cite{curl}, for the particular case of the operator curl, that there are extensive cancellations taking place when taking the matrix trace, so that the resulting asymmetry operator is, in fact of order $-3$. Having an operator of negative order at one's disposal, one can hope to define geometric invariants capturing spectral asymmetry by a regularisation procedure. It was proved in \cite{conjectures} that the regularised trace of the asymmetry operator for curl yields precisely the classical eta invariant for curl.

The general strategy underpinning the above results is not specific to the operator curl, and one can deploy it more widely to study spectral asymmetry for non-semibounded systems of PDEs (or, more generally, pseoudodifferential matrix operators). There are several advantages to this approach. For instance, rather than `just' a number --- the eta invariant --- (or a function, if one considers the local version thereof), one obtains an invariantly defined operator, which encodes more information about spectral asymmetry and its relationship with the geometry of the underlying manifold. Furthermore, the above approach is `direct' and, in a sense, intuitive, thus potentially more accessible to the more applied side of the mathematical community.

In this paper, with also a mathematical physics readership in mind, we will show how one can apply the techniques developed in \cite{curl,conjectures} to the important and physically meaningful case of the massless Dirac operator. When compared to what happens for curl, on the one hand the massless Dirac operator is easier to deal with, because one does not need to work with integral kernels and symbols that are two-point tensors, alongside the mathematical hurdles that come with that: these quantities are now matrix functions valued in half-densities. On the other hand, one faces new technical challenges, such as the fact that one needs to account for gauge transformations connecting spectrally equivalent framings, as well as the fact that the asymmetry operator for massless Dirac is `just' of order $-1$ (as opposed to $-3$).

At the same time, from a physical viewpoint, studying the spectral properties of the Dirac operator on a three dimensional Riemannian manifold is of great relevance due to its several applications, especially in connection with the eta invariant (we shall elaborate further on this towards the end of this section). Listing them all goes beyond the scope of this work, but we feel it is worth emphasising at least two notable examples, which are of great relevance in high energy and condensed matter physics. As to the former field, the eta-invariant enters in the study of the path-integral formulation of a Chern--Simons theory \cite{Witten:1988hf}; as to the latter field, the last decade has witnessed an increasing interest in the analysis of topological phases of matter, see for example \cite{Witten:2015aoa} for a review. In this regard, Weyl Fermions and their spectral property become a crucial ingredient in many instances, such as the study of the quasiparticle excitations of the so-called Weyl semimetals.

\

The remainder of this section is devoted to the statements of our main results, postponing their proofs and further explanations until later sections. 

\

In order to obtain local and global geometric invariants from the asymmetry operator $A$, it is convenient to work with a distinguished choice of framing: the \emph{Levi-Civita framing} --- see also \cite[Section 7.1]{dirac}. 

\begin{definition}
	\label{def: Levi-Civita framing}
	Given $z \in M$, let $\mathcal{U}_z$ be a geodesic neighbourhood of $z$. For $x \in \mathcal{U}_z$ and $j\in\{1,2,3\}$, let $\tilde{e}_j^{\mathrm{loc}}(x)$ be the parallel transport of $e_j(z)$ along the unique geodesic connecting $z$ to $x$. The \emph{Levi-Civita framing} generated by $\{e_j\}_{j=1}^3$ at $z$ is defined as the equivalence class of framings which coincide with $\{\tilde{e}_j^{\mathrm{loc}}\}_{j=1}^3$ in a geodesic neighborhood of $z$. 
\end{definition}

\noindent Let us denote by $\mathfrak{a}(x,y)$ the Schwartz kernel of $A$ and by
\[
a(x,\xi)\sim \sum_{j=0}^{+\infty} a_{-j}(x,\xi), \qquad a_{-j}(x,\lambda\xi)=\lambda^{-j} a_{-j}(x,\lambda\xi) \quad \forall\lambda>0,
\]
its left symbol. Here $\sim$ stands for asymptotic expansion ``in smoothness'' \cite[\S~3.3]{shubin}.

\

Let us fix a point $z\in M$ and let us denote by $\tilde W$ the Dirac operator associated with the Levi-Civita framing $\{\tilde e_j\}_{j=1}^3$ generated by $\{e_j\}_{j=1}^3$ at $z$. Let $\tilde{\mathfrak{a}}$ and $\tilde{a}$ be the Schwartz kernel and full left symbol, respectively, of the corresponding asymmetry operator.

\begin{theorem}
\label{proposition symbol of tilde A}
We have
\begin{equation}
\label{proposition symbol of tilde A equation 1}
\tilde{a}_j(z,\xi)=0 \qquad \text{for }\ j=0,-1,-2,
\end{equation}
and
\begin{equation}
\label{proposition symbol of tilde A equation 2}
\tilde A_\prin(z,\xi):=\tilde{a}_{-3}(z,\xi)=- \frac{1}{12} E^{\alpha \, \beta \, \gamma} (z) \nabla_{\alpha} \operatorname{Ric}_{\beta}{}^{\rho} (z) \frac{\xi_{\gamma} \xi_{\rho}}{|\xi|^5}\,,
\end{equation}
where 
\begin{equation}
\label{Eabc}
E_{\alpha\beta\gamma}(x):=\rho(x)\,\varepsilon_{\alpha\beta\gamma}
\end{equation}
is the Levi-Civita tensor, $\rho$ the Riemannian density, while $\varepsilon_{\alpha \beta \gamma}$ is the totally antisymmetric symbol, $\varepsilon_{123}=+1$.
\end{theorem}

\begin{remark}
Let us compare the content of Theorem~\ref{proposition symbol of tilde A} with the corresponding results for the operator curl. Formulae~\cite[Eqn.~(1.13)]{baer_curl} and~\eqref{proposition symbol of tilde A equation 2} tell us that the principal symbol of the asymmetry operator for curl and Dirac (in the sense of the above theorem) possess the same structure, and differ only by a factor $6$. Note that the two operators possess very similar asymptotic distributions of eigenvalues, as captured by Weyl-type laws: compare~\cite[Theorem~3.6 and Remark~3.7]{baer_curl} with~\cite[Remark~1.5]{dirac}. In fact, similarities between the two operators extend to higher order terms in the asymptotic expansion for the eigenvalue counting function, see~\cite{5terms}.
\end{remark}

\begin{proposition}
\label{prop: disc plus cont}
Fix $x\in M$. For every $y$ in a geodesic neighbourhood of $x$ we have
\begin{equation}\label{Eq: decomposition}
\tilde{\mathfrak{a}}(x,y)=
\frac{1}{72\pi^2} E^{\alpha\beta}{}_\gamma (x) \nabla_{\alpha} \operatorname{Ric}_{\beta\rho} (x) \frac{(x-y)^\gamma (x-y)^\rho}{\dist^2(x,y)}
+
\tilde{\mathfrak{a}}_c(x,y)\,,
\end{equation}
where $\tilde{\mathfrak{a}}_c(x,y)$ is continuous at $x$ in the variable $y$.
\end{proposition}

Let us immediately note that formula~\eqref{Eq: decomposition} does depend on the choice of local coordinates and we shall elaborate more on this in Section \ref{The regularised trace of the asymmetry operator}. This notwithstanding, the above decomposition allows us to introduce the following definition, which, combined with the ensuing theorem, is the main result of this paper.

\begin{definition}
\label{def: regularised local trace}
We define the \emph{regularised local trace} of the asymmetry operator $A$ as
\begin{equation}\label{Eq: psiDirloc}
\psi_\mathrm{Dir}^\mathrm{loc}(x):=\tilde{\mathfrak{a}}_c(x,x)
\end{equation}
and the \emph{regularised global trace} of the asymmetry operator $A$ as
\end{definition}
\begin{equation}\label{Eq: psiDir}
\psi_\mathrm{Dir}:=\int_M \psi_\mathrm{Dir}^\mathrm{loc}(x) \,\rho(x)\,\mathrm{d}x\,.
\end{equation}

\begin{theorem}
\label{main theorem regularised local trace well defined}
The map $x\mapsto \psi_\mathrm{Dir}^\mathrm{loc}(x)$ defined in accordance with~\eqref{Eq: psiDirloc} is a well defined scalar function, independent of the choice of local coordinates --- a geometric invariant describing the spectral asymmetry of our Dirac operator $W$.
\end{theorem}

A natural question is: are the geometric invariants from Definition~\ref{def: regularised local trace} new? If not, how do they relate to known geometric invariants? We conjecture that the regularised local and global traces coincide precisely with the classical local and global eta invariants for the massless Dirac operator, respectively. 

More explicitly, consider the quantities\footnote{Recall that we adopt the convention $\lambda_k\neq 0$.}
\begin{equation}
    \label{local eta invariant}
    \eta_{\mathrm{Dir}}^\mathrm{loc}(s;x):=\sum_{k}\frac{\mathrm{sgn}\lambda_k}{|\lambda_k|^s} \,v^*_k(x)\, v_k(x)
\end{equation}
and
\begin{equation}
    \label{global eta invariant}
     \eta_{\mathrm{Dir}}(s):=\sum_{k}\frac{\mathrm{sgn}\lambda_k}{|\lambda_k|^s}\,,
\end{equation}
which are traditionally known as the local and global eta functions for the massless Dirac operator, respectively. It is not hard to show that \eqref{global eta invariant} is absolutely convergent for $\operatorname{Re}s>3$ and it admits a meromorphic continuation to the whole complex plane with possible first order poles at $s\in \mathbb{Z}\cap (-\infty,3]$. The same is true for~\eqref{local eta invariant} for any given $x\in M$. Of course, the local and global eta functions are related as
$\eta_{\mathrm{Dir}}(s)=\int_M \eta_{\mathrm{Dir}}^\mathrm{loc}(s;x)\,\rho(x)\,\mathrm{d}x$. Although we should like to refrain from delving too deep into the literature on the topic (which is incredibly extensive), we cannot avoid mentioning that the eta function is, in fact, holomorphic at $s=0$; the number $\eta(0)$ --- the \emph{eta invariant} --- is the commonly accepted measure of spectral asymmetry. We refer the reader to~\cite{asymm1,asymm2, asymm3, asymm4, hitchin} as foundational references on the subject.

We envisage that
\begin{equation}
    \psi_\mathrm{Dir}^\mathrm{loc}(x)=\eta_{\mathrm{Dir}}^\mathrm{loc}(0;x), \qquad  \quad\psi_\mathrm{Dir}=\eta_{\mathrm{Dir}}(0).
\end{equation}
Whilst providing a fully fledged proof of this correspondence is beyond the scope of the present paper, preliminary calculations suggest that arguments in the spirit of \cite{conjectures} would yield the desired result. We postpone a detailed analysis of this matter until future work.

\subsection*{Notation}
\addcontentsline{toc}{subsection}{Notation}

\begin{longtable}{l l}
\hline
\\ [-1em]
\multicolumn{1}{c}{\textbf{Symbol}} & 
  \multicolumn{1}{c}{\textbf{Description}} \\ \\ [-1em]
 \hline \hline \\ [-1em]
$\sim$ & Asymptotic expansion \\ \\ [-1em]
$\|\,\cdot\,\|$ & Riemannian norm \\ \\ [-1em]
$|\,\cdot\,|$ & Euclidean norm \\ \\ [-1em]
$\langle\,\cdot\,\rangle$ & Japanese bracket\\ \\ [-1em]
$A$ & Asymmetry operator, Definition~\ref{asyop} \\ \\ [-1em]
$\mathfrak{a}(x,y)$ & Integral kernel of the asymmetry operator $A$ \\ \\ [-1em]
$\dist$ & Geodesic distance \\ \\ [-1em]
$e_j{}^\alpha(x)$, $e^k{}_\beta(x)$ & Framing and dual framing \\ \\ [-1em]
$\tilde e_j{}^\alpha(x)$, $\tilde e^k{}_\beta(x)$ & Levi-Civita framing and dual Levi-Civita framing, Definition~\ref{def: Levi-Civita framing} \\ \\ [-1em]
$\varepsilon_{\alpha\beta\gamma}$ & Totally antisymmetric symbol, $\varepsilon_{123}=+1$ \\ \\ [-1em]
$E_{\alpha\beta\gamma}$ & Totally antisymmetric tensor, Equation \eqref{Eabc} \\ \\ [-1em]
$f_{x^\alpha}$ & Partial derivative of $f$ with respect to $x^\alpha$ \\ \\ [-1em]
$g$ & Riemannian metric \\ \\ [-1em]
$\gamma(z,x;\tau)$ & Geodesic connecting $z$ to $x$, with $\gamma(z,x;0)=z$ and $\gamma(z,x;1)=x$ \\ \\ [-1em]
$\Gamma^\alpha{}_{\beta\gamma}$ & Christoffel symbols\\ \\ [-1em]
$\eta_{\text{Dir}}^\loc(s;x)$ & Local eta function of the operator $W$, Equation \eqref{local eta invariant} \\ \\ [-1em]
$\eta_{\text{Dir}}(s)$ & Eta function of the operator $W$, Equation \eqref{global eta invariant} \\ \\ [-1em]
$\text{H}^m(M)$ & Generalisation of the usual Sobolev spaces $\text{H}^m$ on coloumns of half-densities \\ \\ [-1em]
$\operatorname{I}$ & Identity matrix \\ \\ [-1em]
$\operatorname{Id}$ & Identity operator \\ \\ [-1em]
$(\lambda_j, u_j)$, $j=\pm1, \pm2, \ldots$ & Eigensystem for $W$ \\ \\ [-1em]
$M$ & Connected oriented closed manifold\\ \\ [-1em]
$\operatorname{mod} \ \Psi^{-\infty}$ & Modulo an integral operator with infinitely smooth kernel \\ \\ [-1em]
$P_0$, $P_\pm$ & Equations \eqref{P_+}, \eqref{P_-} and \eqref{P_0} \\ \\ [-1em]
$p_\pm(x,\xi)$ & Full symbol of $P_\pm$ \\ \\ [-1em]
$Q_\prin$ & Principal symbol of the pseudodifferential operator $Q$ \\ \\ [-1em]
$Q_{\prin,s}$ & Principal symbol of $Q$, a pseudodifferential operator of order $-s$ \\ \\ [-1em]
$Q_{\sub}$ & Subprincipal symbol of $Q$, for operators on coloumns of half-densities  \\ \\ [-1em]
$q(x, \xi)$ & Full symbol of $Q$  \\ \\ [-1em]
$\mathfrak{q}(x, y)$ & Schwartz (integral) kernel of $Q$  \\ \\ [-1em]
$\Riem$, $\Ric$, $\text{Sc}$ & Riemann curvature tensor, Ricci tensor and scalar curvature \\ \\ [-1em]
$\rho(x)$ & Riemannian density \\ \\ [-1em]
$\mathbb{S}_r(x)$ & Geodesic sphere of radius $r$ centred at $x\in M$\\ \\ [-1em]
$\mathfrak{tr}$ & Matrix trace, Equation \eqref{matrixtr} \\ \\ [-1em]
$\operatorname{Tr}$ & Operator trace, Equation \eqref{optrace}  \\ \\ [-1em]
$TM$, $T^*M$ & Tangent and cotangent bundle \\ \\ [-1em]
$\psi_{\text{Dir}}^\loc(x)$ & Regularised local trace of $A$, Equation \eqref{Eq: psiDirloc} \\ \\ [-1em]
$\psi_{\text{Dir}}$ & Regularised global trace of $A$, Equation \eqref{Eq: psiDir} \\ \\ [-1em]
$\Psi^m$ & Classical pseudodifferential operators of order $m$ \\ \\ [-1em]
$\Psi^{-\infty}$ & Infinitely smoothing operators, Equation $\eqref{smoothop}$ \\ \\ [-1em]
\hline
\end{longtable}

\section{Positive and negative spectral projections}
\label{Positive and negative spectral projections}

The proof of Theorem~\ref{proposition symbol of tilde A} relies on cancellations taking place at the level of the symbols $p_\pm(x,\xi)$ of the pseudodifferential projections $P_{\pm}$ (recall~\eqref{definition asymmetry operator equation}).

\

\noindent Let us denote by 
\begin{equation*}
    (P_{\pm})_{\text{prin}} (x, \xi) =: P^{(\pm)}(x, \xi),
\end{equation*}
the principal symbols of $P_{\pm}$, which read 
\begin{equation}
\label{P^pm}
    P^{(\pm)} (x, \xi) = \frac{1}{2} \left(\text{I} \pm \frac{W_{\prin}(x, \xi)}{h(x, \xi)} \right).
\end{equation}
Of course, the quantities \eqref{P^pm} are the eigenprojections of \eqref{Wprin} corresponding to the eigenvalues $h^{(\pm)}(x,\xi)$.
%
%
%
%
%
The full symbols of the pseudodifferential projections $P_{\pm}$ can be constructed explicitly by means of an algorithmic procedure described in~\cite[Section~4.3]{part1}. We recall below, in an abridged and simplified fashion, the key steps thereof.

\begin{itemize}
    \item[1.] \textbf{Step 1}: Choose two pseudodifferential operators of order zero $P_{\pm,0} \in \Psi^0(M)$ such that 
    \begin{equation*}
        (P_{\pm,0})_{\text{prin}} = P^{(\pm)}\,,
    \end{equation*}
    where $P^{(\pm)}$ given by~\eqref{P^pm}.
    
    \item[2.] \textbf{Step 2}: For $k=1,2,\dots$, define 
    \begin{equation*}
        P_{\pm,k} := P_{\pm,0} + \sum_{l=1}^k X_{\pm,l}\,,
    \end{equation*}
    where $X_{\pm, l} \in \Psi^{-l}(M)$ are pseudodifferential operators of order $-l$ determined by the following recursive procedure. 
    \item[3.] \textbf{Step 3}: Assuming one has determined $P_{\pm, k-1}$, compute the following quantities:
    \begin{itemize}
        \item[(a)] $R_{\pm,k} := - \left((P_{\pm,k-1})^2 - P_{\pm,k-1}\right)_{\text{prin}, k}\,$, 
        \item[(b)] $S_{\pm,k} := - R_{\pm,k} + P^{(\pm)} R_{\pm,k} + R_{\pm,k} P^{(\pm)}$, 
        \item[(c)] $T_{\pm,k} := [P_{\pm,k-1}, W]_{\text{prin}, k-m} + [S_{\pm,k}, W_{\text{prin}}]\,$, 
    \end{itemize}
    where $[\,\cdot\,,\,\cdot\,]$ denotes the commutator. The pseudodifferential operators $X_{\pm,k} \in \Psi^{-k}(M)$ are then chosen in such a way that their principal symbols read
    \begin{equation}
    \label{Eq: Xpmk}
        (X_{\pm,k})_{\text{prin}} = S_{\pm,k} \pm \frac{1}{2h}\left(P^{(\pm)} T_{\pm,k} P^{(\mp)} - P^{(\mp)} T_{\pm,k} P^{(\pm)}\right)\,.
    \end{equation}
    \item[4.] \textbf{Step 4}: Set 
    \begin{equation}
    \label{P_+-}
        P_\pm \sim P_{\pm,0} + \sum_{l=1}^{\infty} X_{\pm,l},
    \end{equation}
    where the symbol $\sim$ denotes an asymptotic expansion in smoothness, namely, if one truncates the series appearing in formula~\eqref{P_+-} after $k$ terms, one obtains an approximation of $P_{\pm}$ modulo a pseudodifferential operator of order $-k-1$. 
\end{itemize}

\begin{remark}   
\label{remark2.1}
    Note that one can choose the pseudodifferential operators $P_{\pm,0}$ by additionally requiring that the sub-leading homogeneous component of their symbols vanish identically, that is
    \begin{equation*}
        [p_{\pm, 0}(x,\xi)]_{-1} = 0.
    \end{equation*}
    This is convenient in that it implies
    \begin{equation*}
        [p_{\pm,1}]_{-1}=(X_{\pm, 1})_{\text{prin}}.
    \end{equation*}
    Thus, the task of determining $P_{\pm,1}$ reduces to the computation of $(X_{\pm, 1})_{\text{prin}}$. This argument can be iterated at each step of the above algorithm, and will help to reduce the computational complexity of the calculations in the upcoming sections.   
\end{remark}

\section{The asymmetry operator}
\label{The asymmetry operator}

In this section we will study the asymmetry operator $A$ and prepare the ground to demonstrate, in the next section, that one can compute a regularised trace thereof, thus obtaining our local and global geometric invariants and proving our main result.

\subsection{The Dirac operator and gauge transformations}
\label{The Dirac operator and gauge transformations}

Definition~\ref{masslessdirop} depends manifestly on the choice of framing $\{e_j\}_{j=1}^3$. However, the \emph{spectrum} of $W$ does not. In this subsection, we will briefly elaborate on this matter, clarifying some aspects that will be relevant later on.

Suppose we are given a second framing $\{\widetilde{e}_j\}_{j=1}^3$ with the same orientation as $\{e_j\}_{j=1}^3$, and let us denote by $\widetilde{W}$ the corresponding massless Dirac operator. Then there exists a smooth matrix-valued function $O \in C^{\infty}(M; SO(3))$ such that the two framings are related as
\begin{equation}
\label{fram2}
    e_j := O_j{}^k \, \widetilde{e}_k, \qquad j = 1, 2, 3.
\end{equation}
Since $SU(2)$ is the double cover of $SO(3)$, the rotation $O$ of the framing corresponds, at the level of the operator, to a gauge transformation $G \in C^{\infty}(M; SU(2))$ connecting $\widetilde{W}$ and $W$ as
\begin{equation}
\label{fram3}
    W:= G^* \widetilde{W} G, 
\end{equation}
where $^*$ denotes Hermitian conjugation. 
Formula \eqref{fram3} implies\footnote{We denoted with a tilde quantities associated with the `new' framing $\{\widetilde{e}_j\}_{j=1}^3$.}
\begin{equation}\label{Eq: Gauge Transformed Dirac}
   W = - i G^* \widetilde{\sigma}^{\alpha} G \partial_{\alpha} - i G^*\widetilde{\sigma}^{\alpha}   (\partial_{\alpha} G)+ G^*\widetilde{W}_0 G\,, 
\end{equation}
from which it readily follows that
\begin{equation}
    W_{\prin} = G^* \widetilde{W}_{\prin} G\,, \qquad W_0 = - i G^*\widetilde{\sigma}^{\alpha}   (\partial_{\alpha} G) + G^*\widetilde{W}_0 G\,.
\end{equation}
Furthermore, we have
\begin{flalign}
    P^{(\pm)} = G^* \widetilde{P}^{(\pm)} G\,.
\end{flalign}

Now, if $u_k$ is an eigenspinor of the massless Dirac operator $W$ corresponding to the eigenvalue $\lambda_k$, then $\widetilde{u}_k := G u_k$ is an eigenspinor of $\widetilde{W}$ corresponding to the same eigenvalue, as one can easily see from~\eqref{Eq: Gauge Transformed Dirac}. This entails, on the one hand, that the $P_\pm$ behaves covariantly under a gauge transformation. On the other hand, it guarantees that, when we consider the asymmetry operator~\eqref{definition asymmetry operator equation}, if a regularised operator trace is constructed via a procedure which is covariant under gauge transformations, then cyclicity of the trace entails independence of the outcome from $G(x)$. We will rely on this simple fact later on in the paper. 

\subsection{Local calculations for the Levi-Civita framing}
\label{Local calculations with the Levi-Civita framing}

The symbol of the asymmetry operator~\eqref{definition asymmetry operator equation} can be expresses in terms of two types of geometric invariants: torsion of the Weitzenb\"ock connection associated with framing (see, e.g.,~\cite[Appendix~A]{dirac}) and curvature of the Levi-Civita connection. Now, it is not hard to check that, for a generic choice of framing, the asymmetry operator $A$ is of order $-1$, with principal symbol expressed in terms of torsion. Hence, unlike in the case of curl examined in \cite{curl}, $A$ appears to be quite far from being trace-class (which, let us recall, would correspond to the order of the operator being strictly less than $-3$), and thus the strategy from ~\cite{curl, conjectures} seems at first glance rather challenging from the computational point of view.

However, as discussed in the previous subsection, spectral invariants do \emph{not} depend on the framing; therefore, all contributions containing torsion must eventually cancel each other out. On the basis of this observation, in the following we will fix a point $z\in M$ and perform calculations locally, in a neighbourhood of $z\in M$, for the particular choice of the Levi-Civita framing based at $z$ (recall Definition~\ref{def: Levi-Civita framing}). The advantage in this is that the torsion of the Weitzenb\"ock connection for the Levi-Civita framing based at $z$ vanishes at $z$ \cite[App. A]{dirac} --- a property that we can use to our advantage in the calculations. The price one pays in pursuing this strategy is that intermediate steps will no longer be invariant under changes of framing and local coordinates, but the final result will be.

\


Let us fix a point $z\in M$ and choose geodesic normal coordinates centered thereat. Given our global framing $\{e_j\}_{j=1}^3$, we define
\begin{equation}
    V_j := e_j (z) \in T_z M
\end{equation}
and denote by $\{\widetilde{e}_j\}_{j=1}^3$ the Levi-Civita framing generated by $\{e_j\}_{j=1,2,3}$ at $z$, in accordance with Definition~\ref{def: Levi-Civita framing}. In the remainder of this section we will show that \emph{locally}, in the sense of Theorem~\ref{proposition symbol of tilde A}, the asymmetry operator $\tilde A$ associated with $\{\widetilde{e}_j\}_{j=1}^3$ is of order $-3$, thus opening the way for a regularisation procedure of the local trace in the spirit of \cite{curl}.

To this end, let us start by working out explicit expressions for some quantities of interest in the chosen coordinate system.

\begin{proposition}
\label{Prop: framing in normal coordinates}
In the chosen normal coordinate system and for $x$ in a geodesic neighborhood of $z\in M$, the Levi-Civita framing admits the following expansion:
   \begin{flalign}
\label{explevicivfin}
\widetilde{e}_j{}^{\alpha} (x) &= V_j{}^{\alpha} + \frac{1}{6} V_j{}^{\mu}\operatorname{Riem}^{\alpha}{}_{\beta \mu \nu} (z) x^{\beta} x^{\nu} - \frac{1}{6} V_j{}^{\mu} \partial_{\sigma} \partial_{\nu} \Gamma^{\alpha}{}_{\beta \mu} (z) x^{\beta} x^{\nu} x^{\sigma} + O (|x|^4). 
\end{flalign}
\end{proposition}

\begin{proof}
Let $\mathcal{U}_z$ be a geodesic neighborhood of $z$ and let $x \in \mathcal{U}_z$. Let $\gamma: [0,1] \rightarrow M$ be the unique geodesic connecting $z$ to $x$, $\gamma(z, x;0)=z$, $\gamma(z,x;1)=x$, paramterised by arc-length. Now, \cite[Lemma 7.2]{dirac} implies 
\begin{equation}
\label{expframe}
    \widetilde{e}_j{}^{\mu}(tx) = 
     V_j{}^{\mu}
    + \frac{t^2}{6} V_j{}^{\mu}\operatorname{Riem}^{\alpha}{}_{\beta \mu \nu} (z) x^{\beta} x^{\nu}
 + O(|x|^3)\,.
    \end{equation}
By definition, the Levi-Civita framing satisfies
$
    \dot{\widetilde{e}}_j{}^{\alpha} (\gamma(t)) = - \dot{\gamma}^{\beta}(t) \Gamma^{\alpha}{}_{\beta \mu} (\gamma(t)) \widetilde{e}_j{}^{\mu} (\gamma(t))
$
for all $t\in[0,1]$, which, upon integration, yields
\begin{flalign}
\label{framingint}
    \widetilde{e}_j{}^{\alpha} (x) &= V_j{}^{\alpha} - \int_0^1 \Gamma^{\alpha}{}_{\beta \mu} (t x) \, \widetilde{e}_j{}^{\mu}(t x) \, x^{\beta}\,\dr t\,.
\end{flalign}
Plugging~\eqref{expframe} into~\eqref{framingint}, expanding all quantities modulo $O(|x|^4)$ and integrating in $t$, one obtains the claim.
\end{proof}

A straightforward consequence of the previous proposition is the expansion of the projected Pauli matrices and their derivatives in a normal neighbourhood of $z$.

\begin{corollary}
In geodesic normal coordinates $x$ centered at $z \in M$ we have
\begin{flalign}
\label{eq1pauli}
\widetilde{\sigma}^{\alpha}(z) &= \sigma^{\alpha} (z), \\
    \label{eq2pauli}
    \left[\widetilde{\sigma}^{\alpha}\right]_{x^{\beta}} (z) &= 0,\\
    \label{eq3pauli}
    \left[\widetilde{\sigma}^{\alpha}\right]_{x^{\beta}x^{\delta}} (z) &= \frac{1}{6} \left(\mathrm{Riem}^{\alpha}{}_{\beta\mu\delta } (z) + \mathrm{Riem}^{\alpha}{}_{\delta\mu\beta} (z)\right)\sigma^{\mu}(z), \\
    \label{eq4pauli}
\left[\widetilde{\sigma}^{\alpha}\right]_{x^{\beta} x^{\delta}x^{\gamma}} (z) &= - \frac{1}{3} \left( \partial_{\beta} \partial_{\delta} \Gamma^{\alpha}{}_{\gamma \mu} (z) + \partial_{\gamma} \partial_{\beta} \Gamma^{\alpha}{}_{\delta \mu} (z) + \partial_{\delta} \partial_{\gamma} \Gamma^{\alpha}{}_{\beta \mu} (z)\right)\sigma^{\mu}(z).
\end{flalign}
\end{corollary}

\noindent We are now ready to prove our first main result, Theorem~\ref{proposition symbol of tilde A}.

\begin{proof}[Proof of Theorem~\ref{proposition symbol of tilde A}]
For the sake of clarity, we will break the somewhat lengthy proof of~\eqref{proposition symbol of tilde A equation 1} and~\eqref{proposition symbol of tilde A equation 2} into several steps. In the upcoming argument, we shall exploit ideas from~\cite[Section~6]{curl}.

To simplify the calculations, let us introduce some simplifications in the algorithm outlined in Section~\ref{Positive and negative spectral projections}.
Henceforth, we shall label by (a) the set of simplifications relied upon in the proof of~\eqref{proposition symbol of tilde A equation 1} and by (b) those exploited in the computation of~\eqref{proposition symbol of tilde A equation 2}. When not explicitly stated, the simplifications apply to both cases. \\

\textbf{Simplification 1}. 

\noindent Recall that in geodesic normal coordinates centered at $z = 0$ the Riemannian metric reads \cite[Eqn.~(3.4)]{Schoen and Yau}
\begin{equation}
\label{expansion metric}
    g_{\alpha\beta}(x) = \delta_{\alpha\beta}(x) - \frac{1}{3} \operatorname{Riem}_{\alpha\mu\beta\nu}(0)\, x^{\mu} x^{\nu} - \frac{1}{6} \nabla_{\sigma} \operatorname{Riem}_{\alpha\mu \beta\nu}(0)\, x^{\mu} x^{\nu} x^{\sigma} + O(|x|^4). 
\end{equation}
\begin{enumerate}
    \item[(a)] To determine the homogeneous components of order $-1$ and $-2$ of the symbol of $\tilde A$ it suffices to consider \eqref{expansion metric} modulo $O(|x|^3)$. Therefore, we will drop cubic terms in \eqref{expansion metric}.
    \item[(b)]  To determine the homogeneous component of order $-3$ we will retain the cubic terms as well. 
\end{enumerate}

\

\textbf{Simplification 2}.

\begin{itemize}
    \item[(a)] It is not hard to see that $a_{-2}$ will depend linearly on the Ricci tensor $\Ric$. Here we are implicitly using the well known fact that in dimension $d=3$ the Riemann tensor can be expressed in terms of the Ricci tensor via the identity
    \begin{flalign}
\label{Riemtoric}
\operatorname{Riem}_{\alpha\beta\gamma \delta}(x) &= \operatorname{Ric}_{\alpha\gamma}(x) g_{\beta\delta}(x) - \operatorname{Ric}_{\alpha \delta}(x) g_{\beta \gamma}(x) + \operatorname{Ric}_{\beta\delta}(x) g_{\alpha\gamma}(x) \\ \notag & - \operatorname{Ric}_{\beta \gamma} (x) g_{\alpha\delta}(x) + \frac{\operatorname{Sc}(x)}{2} \left( g_{\alpha\delta}(x) g_{\beta\gamma}(x) - g_{\alpha\gamma}(x) g_{\beta  \delta}(x)\right)\,.
\end{flalign}
    Therefore, we can set 
\begin{equation}
\label{Ric zero simplification a}
    \Ric (0) = \begin{pmatrix}
        c_1 & c_2 & c_3 \\
        c_2 & c_4 & c_5 \\
        c_3 & c_5 & c_6
    \end{pmatrix}
\end{equation}
and perform all computations taking only one linearly independent component of $\Ric$ at a time to be nonzero. In what follows, we will provide intermediate steps for the particular case $c_1 = 1$, $c_j = 0$ for $j \ne 1$. To prove \eqref{proposition symbol of tilde A equation 1} it then suffices to run the simplified algorithm six times.
\item[(b)] 
It is not hard to see that $a_{-3}$ will be proportional to $\nabla \Riem$.
Therefore, in proving \eqref{proposition symbol of tilde A equation 2} one can assume that 
\begin{equation}
\label{sim2}
    \operatorname{Riem}(0) = 0, \qquad \nabla \operatorname{Riem}(0) \ne 0\,,
\end{equation}
set
\begin{flalign}
\label{ricc0}
    \operatorname{Ric}_{\alpha \, \beta} (0) &= \begin{pmatrix}
        c_1 & c_2 & c_3 \\
        c_2 & c_4 & c_5 \\
        c_3 & c_5 & c_6
    \end{pmatrix} \\
    \label{ricc1}
      \nabla_1 \operatorname{Ric}_{\alpha \, \beta} (0) &= \begin{pmatrix}
        c_7 & c_8 & c_9 \\
        c_8 & c_{10} & c_{11} \\
        c_9 & c_{11} & c_{12}
    \end{pmatrix} \\
    \label{ricc2}
     \nabla_2 \operatorname{Ric}_{\alpha \, \beta} (0) &= \begin{pmatrix}
        c_{13} & c_{14} & c_{15} \\
        c_{14} & c_{16} & c_{17} \\
        c_{15} & c_{17} & c_{18}
    \end{pmatrix} \\
    \label{ricc3}
     \nabla_3 \operatorname{Ric}_{\alpha \, \beta} (0) &= \begin{pmatrix}
        c_{19} & c_{20} & c_{21} \\
        c_{20} & c_{22} & c_{23} \\
        c_{21} & c_{23} & c_{24}
    \end{pmatrix}
\end{flalign}
and perform the computations assuming that only one set of linearly independent components at a time is different from zero. In what follows, we will provide intermediate steps for the particular case
$c_{11} = 1$ and $c_{j} = 0$ for $j \ne 11$, which corresponds to
\begin{equation}
\label{assumingRiclinear}
    \operatorname{Ric}(x) = \begin{pmatrix}
        0 & 0 & 0 \\
        0 & 0 & x^1 \\
        0 & x^1 & 0
    \end{pmatrix}\,.
\end{equation}
Observe that, owing to the Bianchi identities, the number of independent components in \eqref{ricc1}--\eqref{ricc3} is $15$.
\end{itemize}

\textbf{Simplification 3}. 

In view of positive homogeneity and of the fact that the result is rotationally invariant with respect to momentum, without loss of generality one can prove the claim for the particular choice of
\begin{equation}
\label{sim3}
\xi = \begin{pmatrix}
    0 \\ 
    0 \\
    1
\end{pmatrix} + \underbrace{\begin{pmatrix}
    \eta_1 \\
    \eta_2 \\
    \eta_3 
\end{pmatrix}}_{=: \eta},
\end{equation}
and expand the Riemannian norm of $\xi$, $\|\xi\| := \sqrt{g^{\mu \, \nu}(x) \xi_{\mu} \xi_{\nu}}$ in powers of $x$ and $\eta$. 
As soon as there are no remaining partial derivatives with respect to momentum, one can safely evaluate all quantities at $\eta = 0$. \\

\textbf{Simplification 4}.

\begin{itemize}
\item[(a)] When doing calculations in degree of homogeneity in momentum $-k$, drop all terms of order higher than $2-k$ \emph{jointly in $x$ and $\eta$}, assuming that both $x$ and $\eta$ are of the same order. 
\item[(b)] When doing calculations in degree of homogeneity in momentum $-k$, drop all terms of order higher than $3-k$ jointly in $x$ and $\eta$.
\end{itemize}

\noindent We are now ready to move on to the actual proof.

Let us begin by proving~\eqref{proposition symbol of tilde A equation 1}. To this avail, let us adopt the set of simplifications labelled by~(a). For the particular choice $c_{j}=\delta_{1j}$ in \eqref{Ric zero simplification a}, the metric tensor reads 
\begin{equation}
\label{metricsim1}
    g (x) = \begin{pmatrix}
        1 - \frac{1}{6} (x^2)^2 - \frac{1}{6} (x^3)^2 & \frac{1}{6} x^1 x^2 & \frac{1}{6} x^1 x^3 \\
        \frac{1}{6} x^1 x^2 & 1 - \frac{1}{6}(x^1)^2 + \frac{1}{6}(x^3)^2 & - \frac{1}{6} x^2 x^3\\
        \frac{1}{6} x^1 x^3 & - \frac{1}{6} x^2 x^3 & 1 - \frac{1}{6} (x^1)^2 + \frac{1}{6}(x^2)^2
    \end{pmatrix} + O(|x|^3),
\end{equation}
and
\begin{equation}
    \rho(x) = 1 - \frac{1}{6} (x^1)^2 + O(|x|^3). 
\end{equation}
Under the assumption $[p_{\pm,0}]_{-1}=0$, the algorithm from Section~\ref{Positive and negative spectral projections} yields 
\begin{equation}
\label{Xpm-1v1}
    (X_{\pm, 1})_{\prin} (x, \xi) = \begin{pmatrix}
        0 & \pm \frac{1}{48} (3i x^1 - x^2) \\
        \pm \frac{1}{48} (3i x^1 + x^2) & 0 
    \end{pmatrix} + O(|x|^2 + |\eta|^2)
\end{equation}
Now, assuming that $[p_{\pm, 1}]_{-2} = 0$ (recall Remark \ref{remark2.1}), we get 
\begin{equation}
\label{Rpmm2}
    R_{\pm, 2} = - [(P_{\pm, 1})^2]_{-2}.
\end{equation}
Given two pseudodifferential operators $A$ and $B$ with left symbols $a$ and $b$, the left symbol of their composition reads
\begin{equation}
\label{compsym}
    \sigma_{BA} \sim \sum_{k=0}^{+\infty} \frac{1}{i^k k!} \frac{\partial^k b}{\partial \xi_{\alpha_1} ... \partial \xi_{\alpha_k}} \frac{\partial^k a}{\partial x^{\alpha_1} ... \partial x^{\alpha_k}}\,.
\end{equation} 
Formulae \eqref{Rpmm2} and \eqref{compsym} imply
\begin{equation}
    R_{\pm, 2} (0, \xi_0) = \begin{pmatrix}
        - \frac{1}{48} & 0 \\
        0 & - \frac{1}{48}
    \end{pmatrix}.
\end{equation}
Furthermore, we have
\begin{flalign}
    S_{\pm, 2} (0, \xi_0) &= \begin{pmatrix}
        \mp \frac{1}{48} & 0 \\
        0 & \pm \frac{1}{48} 
    \end{pmatrix},  \\
    T_{\pm, 2} (0, \xi_0) &= \begin{pmatrix}
        0 & 0 \\
        0 & 0
    \end{pmatrix},
\end{flalign}
so that,
\begin{equation}
\label{Xpm-2v1}
    (X_{\pm, 2})_{\prin} (0, \xi_0) = \begin{pmatrix}
        \mp \frac{1}{48} & 0 \\
        0 & \pm \frac{1}{48} 
    \end{pmatrix}.
\end{equation}
Formulae \eqref{Xpm-1v1} and \eqref{Xpm-2v1} imply
\begin{equation}
    a_{-1}(0,\xi_0)=\text{tr}[(X_{+,1} - X_{-,1})_{\prin}] (0, \xi_0) = 0, \qquad a_{-2}(0,\xi_0)=\text{tr}[(X_{+,2} - X_{-,2})_{\prin}] (0, \xi_0) = 0,
\end{equation}
which concludes the proof of~\eqref{proposition symbol of tilde A equation 1}. 

\

In the same spirit as above, let us move on to the proof of \eqref{proposition symbol of tilde A equation 2}, adopting the simplifications labeled by $(b)$.

Assuming that the Ricci tensor takes the form \eqref{assumingRiclinear}, one gets
\begin{equation}
\label{metricsim}
    g (x) = \begin{pmatrix}
        1 - \frac{x^1 x^2 x^3}{3} & \frac{(x^1)^2 x^3}{6} & \frac{(x^1)^2 x^2}{6} \\
        \frac{(x^1)^2 x^3}{6} & 1 & - \frac{(x^1)^3}{6}\\
        \frac{(x^1)^2 x^2}{6} & - \frac{(x^1)^3}{6} & 1  
    \end{pmatrix} + O(|x|^4),
\end{equation}
which, in turn, yields
\begin{equation}
\label{metricinvsim}
    g^{-1} (x) = \begin{pmatrix}
        1 +\frac{x^1 x^2 x^3}{3} & -\frac{(x^1)^2 x^3}{6} & -\frac{(x^1)^2 x^2}{6} \\
        -\frac{(x^1)^2 x^3}{6} & 1 & \frac{(x^1)^3}{6}\\
        -\frac{(x^1)^2 x^2}{6} & \frac{(x^1)^3}{6} & 1  
    \end{pmatrix} + O(|x|^4)\,,
\end{equation}
\begin{equation*}
    \rho(x) = 1 - \frac{x^1 x^2 x^3}{6} + O(|x|^4)\,,
\end{equation*}
and
\begin{equation}
    \label{normsim}
    \|\xi\| = 1 + \frac{1}{2} \left(\eta_1^2 + \eta_2^2 + 2 \eta_3 -(\eta_1^2 + \eta_2^2) \eta_3\right) + O(|x|^4 + |\eta|^4)\,. 
    \qquad
\end{equation}
It is easy to see that the only non-vanishing (modulo $O(|x|^2)$) independent component of the Riemann tensor is
\begin{flalign*}
    \operatorname{Riem}_{1213} (x) &=  x^1 \,.
\end{flalign*}
Without loss of generality, we can assume $V_j{}^{\alpha} = \delta_j{}^{\alpha}$ (this can always be achieved by a rigid rotation of the coordinate system), so that the Levi-Civita framing admits the following expansion (recall~\eqref{explevicivfin}):
\begin{equation}
    \widetilde{e}_{j}{}^{\alpha}(x) = \begin{pmatrix}
        1 + \frac{x^1 x^2 x^3}{9} & - \frac{(x^1)^2 x^3}{18} & - \frac{(x^1)^2 x^2}{18} \\
        - \frac{(x^1)^2 x^3}{18} & 1 & \frac{(x^1)^3}{18} \\
        - \frac{(x^1)^2 x^2}{18} & \frac{(x^1)^3}{18} & 1 
    \end{pmatrix} + O(|x|^4)\,.
\end{equation}
The projected Pauli matrices read
\begin{flalign*}
    \widetilde{\sigma}^1(x) &= \begin{pmatrix}
        -\frac{(x^1)^2x^2}{18} & 1+ \frac{x^1x^2 x^3}{9} + i \frac{(x^1)^2 x^3}{18} \\
        1+ \frac{x^1x^2 x^3}{9} - i \frac{(x^1)^2 x^3}{18} & \frac{(x^1)^2x^2}{18}
    \end{pmatrix}+ O(|x|^4), \\
    \widetilde{\sigma}^2(x) &= \begin{pmatrix}
        \frac{(x^1)^3}{18} & - \frac{(x^1)^2 x^3}{18} - i \\
        - \frac{(x^1)^2 x^3}{18} + i & -\frac{(x^1)^3}{18}
    \end{pmatrix}+ O(|x|^4), \\
    \widetilde{\sigma}^3(x) &= \begin{pmatrix}
        1 & - \frac{(x^1)^2 x^2}{18} - i \frac{(x^1)^3}{18} \\
        - \frac{(x^1)^2 x^2}{18} + i \frac{(x^1)^3}{18} & -1
    \end{pmatrix}+ O(|x|^4).
\end{flalign*}

\noindent Thus, straightforward calculations yield
\begin{flalign*}
    P^{(\pm)}(x, \xi) = \begin{pmatrix}
      p^{(\pm)}_{11} (x, \xi) &  p^{(\pm)}_{12} (x, \xi) \\
        p^{(\pm)}_{21} (x, \xi) &  p^{(\pm)}_{22} (x, \xi)
    \end{pmatrix} +O(|x|^4+|\eta|^4),
\end{flalign*}
where 
\begin{flalign*}
      p^{(+)}_{11} (x, \xi) &= 1 + \frac{1}{4}  (2 \eta_3 -1) (\eta_1^2+ \eta_2^2), \\
       p^{(+)}_{12} (x, \xi) & = \frac{1}{36} (-9(\eta_1 - i \eta_2) - i(x^1)^3- (x^1)^2 x^2), \\
       p^{(+)}_{21} (x, \xi) &= \frac{1}{36} (i (x^1)^3 -(x^1)^2 x^2 -9 (\eta_1 + i \eta_2) (- 2 + \eta_1^2 + \eta_2^2 + 2 \eta_3 - 2\eta_3^2) ), \\
       p^{(+)}_{22} (x, \xi) &= - \frac{1}{4} (2\eta_3 - 1) (\eta_1^2 + \eta_2^2),
\end{flalign*}
and
\begin{flalign*}
      p^{(-)}_{11} (x, \eta) &= - \frac{1}{4}  (2 \eta_3 -1) (\eta_1^2+ \eta_2^2), \\
       p^{(-)}_{12} (x, \eta) & = \frac{1}{36} (9(\eta_1 - i \eta_2)(\eta_1^2 + \eta_2^2 - 2 \eta_3^2 + 2 \eta_3 - 2) + i (x^1)^3 + (x^1)^2 x^2), \\
       p^{(-)}_{21} (x, \eta) &= \frac{1}{36} (9 (\eta_1 + i \eta_2) (\eta_1^2 + \eta_2^2 + 2 \eta_3 - 2\eta_3^2 -2) -i (x^1)^3 +(x^1)^2 x^2), \\
       p^{(-)}_{22} (x, \eta) &= 1+ \frac{1}{4} (2\eta_3 - 1) (\eta_1^2 + \eta_2^2).
\end{flalign*}


We are now ready to implement the first iteration of our algorithm. 
Arguing as above, straightforward calculations give us
\begin{equation*}
    R_{\pm, 1} (x, \xi) = \begin{pmatrix}
        - \frac{1}{36} x^1 (2x^1 + i x^2) &0 \\
        0 & \frac{1}{36} x^1 (2x^1 - i x^2)
    \end{pmatrix} + O (|x|^3 + |\eta|^3),
\end{equation*}
\begin{equation*}
     S_{\pm, 1} (x, \xi) = \begin{pmatrix}
        \mp \frac{1}{36} x^1 (2x^1 + i x^2) &0 \\
        0 & \pm \frac{1}{36} x^1 (-2x^1 + i x^2)
    \end{pmatrix}  + O (|x|^3 + |\eta|^3),
\end{equation*}
and
\begin{equation*}
    T_{\pm, 1}(x, \xi) = \begin{pmatrix}
        0 & \pm \frac{1}{72} (7 x^1 - 5i x^2) x^3 \\
        \pm \frac{1}{72} (7x^1 + 5i x^2) x^3 & 0
    \end{pmatrix}  + O (|x|^3 + |\eta|^3).
\end{equation*}
Plugging the above expressions into~\eqref{Eq: Xpmk}, we get
\begin{equation*}
    (X_{\pm, 1})_{\text{prin}}(x, \xi) = \begin{pmatrix}
        \mp \frac{1}{36} x^1 (2x^1 + i x^2) & \pm \frac{1}{144} (7x^1 - 5i x^2) x^3 \\ 
        \pm \frac{1}{144} (7x^1 + 5i x^2) x^3 & \pm \frac{1}{36} x^1 (-2 x^1 + i x^2)
    \end{pmatrix}  + O (|x|^3 + |\eta|^3).
\end{equation*}
This yields
\begin{equation*}
    \text{tr}[(X_{+, 1} - X_{-, 1})_{\text{prin}}] (x, \xi) = - \frac{2}{9} (x^1)^2 + O (|x|^3+ |\eta|^3),
\end{equation*}
which vanishes at $x=z=0$.

Implementing the second iteration of the algorithm under the assumption $[p_{\pm, 1}(x, \xi)]_{-2} = 0$, we get
\begin{equation*}
    R_{\pm, 2} (x, \xi) = -[(P_{\pm, 1})^2 - P_{\pm, 1}]_{\prin, 2} = -[(P_{\pm, 1})^2]_{-2}= \begin{pmatrix}
        -\frac{i}{24} x^3 & 0 \\
        0 & \frac{i}{24} x^3
    \end{pmatrix} + O(|x|^2 + |\eta|^2),
\end{equation*}
\begin{flalign*}
S_{\pm, 2} (x, \xi) 
&= 
\begin{pmatrix}
    \mp \frac{i}{24} x^3 & 0 \\
        0 & \mp \frac{i}{24} x^3
\end{pmatrix}  + O(|x|^2 + |\eta|^2),
\end{flalign*}
\begin{equation*}
     T_{\pm, 2} (x, \xi) = \begin{pmatrix}
        0& \mp \frac{5}{72} i (x^1 \mp i x^2) \\
        \frac{5}{72} (\mp i x^1 + x^2) & 0
    \end{pmatrix}  + O(|x|^2 + |\eta|^2),
\end{equation*}
and thus
\begin{equation}
\label{Xpm2}
     (X_{\pm, 2})_{\text{prin}} (x, \xi) = \begin{pmatrix}
         \mp \frac{i}{24} x^3 & \mp \frac{5}{144} i (x^1 \mp i x^2) \\
         \pm \frac{5}{144} i (x^1 + i x^2) &  \mp \frac{i}{24} x^3
     \end{pmatrix}  + O(|x|^2 + |\eta|^2). 
\end{equation}

\noindent Formula~\eqref{Xpm2} implies
\begin{equation*}
    \tilde{a}_{-2}(x, \xi) := \text{tr}[(X_{+, 2} - X_{-, 2})_{\text{prin}}] (x, \xi) = - \frac{i}{6} x^3 + O(|x|^2 + |\eta|^2), 
\end{equation*}
which vanishes at $x=z=0$.

Assuming $[p_{\pm, 2}(x,\xi)]_{-3} = 0$ and implementing the algorithm a third and final time, we obtain
\begin{equation*}
    R_{\pm, 3} (0, \xi_0) = \begin{pmatrix}
        -\frac{1}{48} & 0 \\
        0 & \frac{1}{48}
        \end{pmatrix},
   \qquad
S_{\pm, 3} (0, \xi_0) =
\begin{pmatrix}
     \mp \frac{1}{48} & 0 \\
        0 & \mp \frac{1}{48}  
\end{pmatrix},
\end{equation*}
\begin{equation*}
     T_{+, 3} (0, \xi_0) = \begin{pmatrix}
        0& 0 \\
       0 & 0
    \end{pmatrix}, 
    \qquad 
    T_{-,3} (0, \xi_0) = \begin{pmatrix}
        \frac{2}{9} & 0 \\
        0 & - \frac{2}{9}
    \end{pmatrix},
\end{equation*}
eventually arriving at
\begin{equation*}
     (X_{\pm, 3})_{\text{prin}} (0, \xi_0) = \begin{pmatrix}
         \mp \frac{1}{48} & 0 \\
         0 & \mp \frac{1}{48}.
     \end{pmatrix}. 
\end{equation*}
The latter finally implies
\begin{equation}
     \tilde{a}_{-3}(0, \xi_0) = \text{tr}[(X_{+, 3} - X_{-, 3})_{\text{prin}}] (0, \xi_0)= - \frac{1}{12}\,,
\end{equation}
which agrees with \eqref{proposition symbol of tilde A equation 2}.

Repeating the above arguments for each set of linearly independent components of the curvature tensor completes the proof. 
\end{proof}

\section{The regularised trace of the asymmetry operator}
\label{The regularised trace of the asymmetry operator}

In this section, relying on the analysis of the asymmetry operator performed in the previous section, we will prove our main result: Theorem~\ref{main theorem regularised local trace well defined}.

\

To this end, let us first give a proof of Proposition~\ref{prop: disc plus cont}, which translates the explicit, local expression~\eqref{proposition symbol of tilde A equation 2} for the symbol of the asymmetry operator into an explicit, local formula for the integral kernel thereof.

\begin{proof}[Proof of Proposition~\ref{prop: disc plus cont}]
Locally, the symbol $\tilde{a}$ and the integral kernel $\tilde{\mathfrak{a}}$ are related as
\begin{equation}
\label{relation symbol vs integral kernel}
    \tilde{\mathfrak{a}}(x,y)=\frac{1}{(2\pi)^3\rho(y)} \int_{\mathbb{R}^3} e^{i(x-y)^\mu \xi_\mu} \tilde{a}(x,\xi)\,\mathrm{d}\xi\,,
\end{equation}
where the identity is understood in a distributional sense. It was shown in \cite[Lemma~7.5]{curl} that
\begin{equation}
\label{fourier transform 1}
\frac{1}{(2\pi)^3}\int_{\mathbb{R}^3} \frac{\xi_\gamma\xi_\rho}{\langle \xi \rangle^5}e^{-i y^\mu \xi_\mu}\,\mathrm{d}\xi= -\frac{1}{12\pi^2}\left[2\frac{y^\mu y^\nu}{|y|^2}+(1+2\ln|y|)\delta_{\gamma\rho}\right]+ h_{\gamma\rho}(y),
\end{equation}
where $\langle \xi \rangle:=(1+|\xi|^2)^{\frac{1}{2}}$ and $h_{\gamma\rho}\in C^1(\mathbb{R}^3)$.

Now, we observe that, owing to the symmetries of the totally antisymmetric symbol, we have
\begin{equation}
\label{symmetric part of principal symbol zero}
    E^{\alpha \beta \gamma} (x) \nabla_{\alpha} \operatorname{Ric}_{\beta}{}^{\rho} (x)\delta_{\gamma\rho}=0\,.
\end{equation}
Therefore, the claim follows from Theorem~\ref{proposition symbol of tilde A}, \eqref{fourier transform 1} and~\eqref{symmetric part of principal symbol zero}, upon observing that
\begin{equation}
\frac1{|\xi|^{5}}-\frac{1}{\langle \xi \rangle^{5}}=O(|\xi|^{-7}) \quad \text{as}\quad |\xi|\to+\infty.
\end{equation}
\end{proof}

\begin{remark}
\label{remark alternative representation local invariant}
Note that our local geometric invariant can also be obtained regularising by means of `averaging over spheres'. Indeed, one has that
\begin{equation}
\label{remark alternative representation local invariant equation}
\psi_\mathrm{Dir}^\loc(x)=\lim_{r\to 0^+}\frac{1}{4\pi r^2}\int_{\mathbb{S}_r(x)} \tilde{\mathfrak{a}}(x,y)\, \mathrm{d}S_y\,,
\end{equation}
where $\mathbb{S}_r(x)$ is the geodesic sphere centred at $x$ of radius $r>0$. 
To see this, one observes that $\int_{\mathbb{S}_r(x)}(x-y)^\alpha(x-y)^\beta\, \mathrm{d}S_y=\frac{4\pi r^2}{3}\delta^{\alpha\beta}$, so that
\[
\frac{1}{4\pi r^2}\int_{\mathbb{S}_r(x)} \tilde{\mathfrak{a}}(x,y)\, \mathrm{d}S_y=
\frac{1}{4\pi r^2}\int_{\mathbb{S}_r(x)} \tilde{\mathfrak{a}}_c(x,y)\, \mathrm{d}S_y
+\frac{1}{216\pi^2} E^{\alpha\beta\rho} (x) \nabla_{\alpha} \operatorname{Ric}_{\beta\rho} (x)\,.
\]
But then the continuity of $\tilde{\mathfrak{a}}_c(x,y)$ at $y$ and~\eqref{symmetric part of principal symbol zero} imply~\eqref{remark alternative representation local invariant equation}.
\end{remark}

\begin{proof}[Proof of Theorem~\ref{main theorem regularised local trace well defined}]
The decomposition~\eqref{Eq: decomposition} does depend on the choice of local coordinates. Indeed, the issue at hand is that $(x-y)^\alpha$ is \emph{not} a vectorial quantity, namely, it does not transform as a vector under changes of local coordinates. However, the value of the continuum component along the diagonal $x=y$ is a scalar quantity, independent of the choice of local coordinates. This can be shown by arguing as in~\cite[Subsection~7.2]{curl}.
\end{proof}

\section*{Acknowledgements}
\addcontentsline{toc}{section}{Acknowledgements}

We are grateful to Dmitri Vassiliev for stimulating discussions on this and related topics.  INFN's support for MC's visit to Pavia, where part of this work has been done, is gratefully acknowledged. MC would also like to thank the Department of Mathematics of Yale University, where this manuscript was finalised during MC's stay as a Visiting Professor, for the kind hospitality.

\

\noindent \emph{Funding}: MC was supported by EPSRC Fellowship EP/X01021X/1, and is a member of the GNAMPA group of INdAM. 
The work of BC has been supported by a fellowship of the University of Pavia and BC is grateful to the Department of Mathematics at Heriot-Watt University for the kind hospitality during the realisation of part of this work. Part of this work has appeared in the Master thesis "On the Spectral Asymmetry of the Dirac Operator on Three-dimensional, Closed, Riemannian Manifolds" submitted on the 10/05/24 as a partial fulfillment of the requirements to obtain a Master degree at IUSS Pavia. 
BC and CD acknowledge the support both of the INFN Sezione di Pavia and of Gruppo Nazionale di Fisica Matematica, part of INdAM.

\vskip.2cm

\noindent\textbf{Data availability statement}. Data sharing is not applicable to this article as no new data were created or analysed in this study.

\vskip .2cm

\noindent\textbf{Conflict of interest statement.} The authors certify that they have no affiliations with or involvement in any
organization or entity with any financial interest or non-financial interest in the subject matter discussed in
this manuscript.





\end{document}